\documentclass[copyright,creativecommons, fleqn]{eptcs}
\providecommand{\event}{GandALF 2014} 
\usepackage{breakurl}             

\usepackage{amsmath}
\usepackage{amssymb}
\usepackage{amsthm}
\usepackage[utf8]{inputenc}
\usepackage{tikz}
\usetikzlibrary{automata}

\newcommand{\quot}[1]{``#1''}

\newtheorem{lemma}{Lemma}
\newtheorem{definition}{Definition}
\newtheorem{theorem}{Theorem}

\newtheorem{corollary}{Corollary}
\newtheorem{example}{Example}

\newcommand{\coloneq}{\mathop{:=}}

\newcommand{\cceq}{\mathop{::=}}

\renewcommand{\epsilon}{\varepsilon}

\newcommand{\pow}[1]{2^{#1}}
\newcommand{\nats}{\mathbb{N}}
\newcommand{\card}[1]{|#1|}
\newcommand{\size}[1]{\card{#1}}
\newcommand{\set}[1]{\{#1\}}

\newcommand{\aut}{\mathfrak{A}}

\newcommand{\arena}{\mathcal{A}}
\newcommand{\game}{\mathcal{G}}
\newcommand{\col}{\Omega}

\newcommand{\ttrue}{\texttt{tt}}
\newcommand{\ffalse}{\texttt{ff}}
\newcommand{\F}{\mathop{\mathbf{F}}}

\newcommand{\U}{\mathbin{\mathbf{U}}}

\newcommand{\conc}{\,;}
\newcommand{\Var}{\mathcal{V}}
\newcommand{\cl}{\mathrm{cl}}
\newcommand{\var}{\mathrm{var}}
\newcommand{\vardiamond}{\mathrm{var}_{\lozenge}}
\newcommand{\varbox}{\mathrm{var}_{\square}}

\newcommand{\ddiamond}[1]{\langle #1 \rangle}
\newcommand{\bbox}[1]{[#1]}

\newcommand{\Rexp}{\mathcal{R}}
\newcommand{\rel}[1]{\mathrm{rel}(#1)}

\newcommand{\ltl}{\mathrm{LTL}}

\newcommand{\pltl}{\mathrm{PLTL}}
\newcommand{\prompt}{\mathrm{PROMPT}$\textendash$\ltl}

\newcommand{\pldl}{\mathrm{PLDL}}
\newcommand{\ldl}{\mathrm{LDL}}
\newcommand{\ldlt}{\mathrm{LDL}_{cp}}
\newcommand{\pldldiamond}{\mathrm{PLDL}_{\lozenge}}
\newcommand{\pldlbox}{\mathrm{PLDL}_{\square}}

\newcommand{\nlogspace}{\textsc{Nlogspace}}

\newcommand{\pspace}{\textbf{\textsc{Pspace}}}
\newcommand{\twoexp}{\textbf{\textsc{2Exptime}}}

\newcommand{\bplus}{\mathcal{B}^{+}}
\newcommand{\trace}{\mathrm{tr}}
\newcommand{\suc}[2]{\mathrm{Succ}_{#1}{#2}}
\newcommand{\outcome}{\mathrm{outcome}}
\newcommand{\trans}{\mathcal{T}}
\newcommand{\sys}{\mathcal{S}}
\newcommand{\marking}{m}

\title{Parametric Linear Dynamic Logic\thanks{This work was partially supported by the German Research Foundation (DFG) as part of SFB/TR 14 ``AVACS''.}}

\author{Peter Faymonville \qquad\qquad Martin Zimmermann 
\institute{Reactive Systems Group, Saarland University, 66123 			Saarbr\"{u}cken, Germany}
\email{\{faymonville, zimmermann\}@react.uni-saarland.de}
}

\def\titlerunning{Parametric Linear Dynamic Logic}
\def\authorrunning{Peter Faymonville and Martin Zimmermann}

\begin{document}
\maketitle


\begin{abstract}
We introduce Parametric Linear Dynamic Logic (PLDL), which extends Linear Dynamic Logic (LDL) by temporal operators equipped with parameters that bound their scope. LDL was proposed as an extension of Linear Temporal Logic (LTL) that is able to express all $\omega$-regular specifications while still maintaining many of LTL's desirable properties like an intuitive syntax and a translation into non-deterministic B\"uchi automata of exponential size. But LDL lacks capabilities to express timing constraints. By adding parameterized operators to LDL, we obtain a logic that is able to express all $\omega$-regular properties and that subsumes parameterized extensions of LTL like Parametric LTL and PROMPT-LTL.

Our main technical contribution is a translation of PLDL formulas into non-deterministic B\"uchi word automata of exponential size via alternating automata. This yields a PSPACE model checking algorithm and a realizability algorithm with doubly-exponential running time. Furthermore, we give tight upper and lower bounds on optimal parameter values for both problems. These results show that PLDL model checking and realizability are not harder than LTL model checking and realizability.
\end{abstract}

\section{Introduction}
\label{sec_intro}
Linear temporal logic ($\ltl$) is a popular specification language for the verification and synthesis of reactive systems. It provides semantic foundations for industrial logics like PSL~\cite{EisnerFismanPSL}. $\ltl$ has a number of desirable properties contributing to its ongoing popularity: it does not rely on the use of variables, it has an intuitive syntax and thus gives a way for practitioners to write declarative and concise specifications. Furthermore, it is expressively equivalent to first-order logic over the natural numbers with successor and order~\cite{Kamp68} and enjoys an exponential compilation property: one can efficiently construct a language-equivalent non-deterministic Büchi automaton of exponential size in the size of the specification.	
The exponential compilation property yields a $\pspace$ model checking algorithm and a $\twoexp$ algorithm for realizability. Both problems are complete for the respective classes.
	
	Model checking of properties described in $\ltl$ or its practical descendants is routinely applied in industrial-sized applications, especially for hardware systems \cite{Forspec02,EisnerFismanPSL}. Due to its complexity, the realizability problem has not reached industrial acceptance (yet). First approaches used a determinization procedure for $\omega$-automata, which is notoriously hard to implement efficiently \cite{AlthoffThomasWallmeier2006}. More recent algorithms for realizability follow a safraless construction \cite{FiliotJinRaskin2011,FinkbeinerSchewe2013}, which avoids explicitly constructing the deterministic automaton, and are showing promise on small examples. 

Despite the desirable properties, two drawbacks of $\ltl$ remain and are tackled by different approaches in the literature: first, $\ltl$ is not able to express all $\omega$-regular properties. For example, the property ``$p$ holds on every even step'' (but may or may not hold on odd steps) is not expressible in $\ltl$, but easily expressible as an $\omega$-regular expression. This drawback is a serious one, since the combination of regular properties and linear-time operators is common in hardware verification languages. Several extensions of $\ltl$~\cite{LeuckerSanchez07, VardiWolper94, Wolper1983} with regular expressions, finite automata, or grammar operators have been proposed as a remedy.  

A second drawback of classic temporal logics like $\ltl$ is the inability to natively express timing constraints. The standard semantics are unable to enforce the fulfillment of eventualities within finite time bounds, e.g., it is impossible to require that requests are granted within a fixed, but arbitrary, amount of time. While it is possible to unroll an a-priori fixed bound for an eventuality into $\ltl$, this requires prior knowledge of the system's granularity and incurs a blow-up when translated to automata, and is thus considered impractical.
 A more practical way of fixing this drawback has been the purpose of a long line of work in parametric temporal logics, such as parametric $\ltl$~\cite{AlurEtessamiLaTorrePeled01}, $\prompt$~\cite{KupfermanPitermanVardi09} and parametric metric interval temporal logic~\cite{GiampaoloTorreNapoli10}. All of them add parameters to the temporal operators to express time bounds, and either test the existence of a global time bound, like $\prompt$, or of individual bounds on the parameters, like parametric $\ltl$. 

Recently, the first drawback was revisited by De Giacomo and Vardi ~\cite{GiacomoVardi13, Vardi11} by introducing an extension of $\ltl$ called linear dynamic logic ($\ldl$), which is as expressive as $\omega$-regular languages. The syntax of $\ldl$ is inspired by propositional dynamic logic (PDL)~\cite{FischerLadner1979}, but the semantics follow linear-time logics. In PDL and $\ldl$, programs are expressed by regular expressions with tests, and temporal requirements are specified by two basic modalities:  $\ddiamond{r}\varphi$ and $\bbox{r}\varphi$, stating that $\varphi$ should hold  at some position where $r$ matches, or at all positions where $r$ matches, respectively. The operators to specify regular expressions from propositional formulas are as follows: sequential composition ($r_1 \conc r_2$), nondeterministic choice ($r_1 + r_2$), repetition ($r^*$), and test $(\varphi?)$ of a temporal formula. On the level of the temporal operators, conjunction and disjunction are allowed. The tests allow to check temporal properties within programs, and are needed to encode $\ltl$ into $\ldl$. 

As an example, the program ``{\bf while} $q$ {\bf do} $a$'' with property $p$ holding after the execution of the loop is expressed in PDL/$\ldl$ as follows:  $\bbox{(q?\conc a)^*\conc \neg q? } p$. Intuitively, the loop condition $q$ is tested on every loop entry, the loop body $a$ is executed/consumed until $\neg q$ holds, and then the post-condition $p$ has to hold. A request-response property (i.e., every request should eventually be followed by a response) can be formalized as follows: $\bbox{\ttrue^*} (\textit{req} \rightarrow \ddiamond{\ttrue^*} \textit{resp})$. 

Both aforementioned drawbacks of $\ltl$, the inability to express all $\omega$-regular properties and the missing capability to specify timing constraints, have been tackled individually in a successful way in previous work, but not at the same time. Here, we propose a logic called $\pldl$ that combines the expressivity of $\ldl$ with the parametricity of P$\ltl$ on infinite traces.

	In $\pldl$, we are for example able to parameterize the eventuality of the  request-response condition, denoted as  ${\bbox{\ttrue^*}(\textit{req}\rightarrow\ddiamond{\ttrue^*}_{\le x}\textit{resp})}$, which states that every request has to be followed by a response within $x$ steps. 
	In the $\pldl$ model checking problem, we determine whether there exists a valuation $\alpha(x)$ for $x$ such that all paths of the system respond to requests within $\alpha(x)$ steps.
	If we take the property as a specification for the $\pldl$ realizability problem, and define \textit{req} as input, \textit{resp} as output, we compute whether there exists a winning strategy that adheres to a valuation $\alpha(x)$ and is able to ensure the delivery of responses to requests in a timely manner.
	 
The main result of this paper is the translation of $\pldl$ to alternating Büchi automata. By an extension of the alternating color technique of \cite{KupfermanPitermanVardi09}, and by very similar algorithms, we obtain the following results: $\pldl$ model checking is $\pspace$-complete and realizability is $\twoexp$-complete. Thus, both problems are no harder than their corresponding variants for $\ltl$. Finally, we give tight exponential and doubly-exponential bounds on satisfying valuations for model checking and realizability.

Our translation might also be of use for $\ldl$ on infinite traces, since De Giacomo and Vardi~\cite{GiacomoVardi13} only considered $\ldl$ on finite traces. Unlike the translation from logic into automata presented there, which is a top-down construction of an alternating automaton, we present  a bottom-up approach.

\section{PLDL}
\label{sec_defs}
Let $\Var$ be an infinite set of variables and let us fix a finite\footnote{This greatly
simplifies our notation and exposition when we translate formulas into automata,
but is not essential.} set~$P$ of atomic propositions which we use to build our
formulas and to label transition systems in which we evaluate them. For a subset $A \in \pow{P}$ and a propositional formula~$\phi$ over $P$, we write $A \models \phi$, if the variable valuation mapping elements in $A$ to true and elements not in $A$ to false satisfies $\phi$. The formulas of $\pldl$ are given by the grammar
\begin{align*}
\varphi &\cceq p \mid \neg p \mid \varphi \wedge \varphi \mid \varphi \vee \varphi
  \mid \ddiamond{r} \varphi 
  \mid \bbox{r} \varphi 
  \mid \ddiamond{r}_{\le z} \varphi 
  \mid \bbox{r}_{\le z} \varphi\\
  r & \cceq \phi \mid \varphi? \mid r+r \mid r \conc r \mid r^*
\end{align*}
where $p \in P$, $z \in \Var$, and where $\phi$ stands for arbitrary propositional formulas over $P$. We use the abbreviations~$\ttrue = p \vee \neg p$ and $\ffalse = p \wedge \neg p$ for some atomic proposition~$p$. The regular expressions have two types of atoms: propositional formulas~$\phi$ over the atomic propositions and tests~$\varphi?$, where $\varphi$ is again a $\pldl$ formula. Note that the semantics of the propositional atom~$\phi$ differ from the semantics of the test~$\phi?$: the former consumes an input letter, while tests do not make progress on the word. This is why both types of atoms are allowed.

The set of subformulas of $\varphi$ is denoted by $\cl(\varphi)$. Note that regular expressions are not subformulas, but the formulas appearing in the tests are, e.g., we have $\cl(\ddiamond{p?\conc q}_{\le x}r) = \set{p, r, \ddiamond{p?\conc q}_{\le x}r}$. The size~$\card{\varphi}$ of $\varphi$ is the sum of $\card{\cl(\varphi)}$ and the sum of the lengths of the regular expressions appearing in $\varphi$ (counted with multiplicity). We define $\vardiamond( \varphi ) = \{ z\in \Var \mid \ddiamond{r}_{\le z}\psi \in
\cl( \varphi) \}$ to be the set of variables parameterizing diamond operators in
$\varphi$, $\varbox( \varphi ) = \{ z\in \Var \mid \bbox{r}_{\le z} \psi \in \cl( \varphi)  \} $
to be the set of variables parameterizing box operators in $\varphi$, and set
$\var( \varphi ) = \vardiamond( \varphi ) \cup \varbox( \varphi )$. Usually, we will denote
variables in $\vardiamond( \varphi )$ by $x$ and variables in $\varbox( \varphi )$ by $y$, if $\varphi$ is clear from the context. A
formula~$\varphi$ is variable-free, if $\var( \varphi ) = \emptyset$.

The semantics of $\pldl$ are defined inductively with respect to an $\omega$-word~$w = w_0 w_1 w_2 \cdots \in (\pow{P})^\omega$, a position~$n \in \nats$, and a variable valuation~$\alpha \colon \Var \rightarrow \nats$ via
\begin{itemize}

\item $(w, n, \alpha) \models p$ if $p \in w_n$ and dually for $\neg p$,


\item $(w, n, \alpha) \models \psi_0 \wedge \psi_1$ if $(w, n, \alpha) \models \psi_0$ and $(w, n, \alpha) \models \psi_1$,

\item $(w, n, \alpha) \models \psi_0 \vee \psi_1$ if $(w, n, \alpha) \models \psi_0$ or $(w, n, \alpha) \models \psi_1$,

\item $(w, n, \alpha) \models \ddiamond{r}\psi$ if there exists $j \in \nats$ s.t.\ $(n, n+j) \in \Rexp(r, w, \alpha)$ and $(w, n+j, \alpha) \models \psi$,

\item $(w, n, \alpha) \models \bbox{r}\psi$ if for all $j \in \nats$ with $(n, n+j) \in \Rexp(r, w, \alpha)$ we have $(w, n+j, \alpha) \models \psi$,

\item $(w, n, \alpha) \models \ddiamond{r}_{\le z}\psi$ if there exists $0 \le j \le \alpha(z)$ s.t.\ $(n, n+j) \in \Rexp(r, w, \alpha)$ and $(w, n+j, \alpha) \models \psi$, 

\item $(w, n, \alpha) \models \bbox{r}_{\le z}\psi$ if for all $0 \le j \le \alpha(z)$ with $(n, n+j) \in \Rexp(r, w, \alpha)$ we have $(w, n+j, \alpha) \models \psi$.

\end{itemize}
Here, the relation~$\Rexp(r,w,\alpha) \subseteq \nats\times\nats$ contains all pairs~$(m,n)$ such that $w_m \cdots w_{n-1}$ matches $r$ ($\alpha$ is needed to evaluate tests in $r$, which might have parameterized subformulas) and is defined inductively by 
\begin{itemize}
\item $\Rexp(\phi,w,\alpha) = \set{(n, n+1) \mid w_n \models \phi}$ for propositional~$\phi$,
\item $\Rexp(\psi?,w,\alpha) = \set{(n, n) \mid (w, n, \alpha) \models \psi}$,
\item $\Rexp(r_0 + r_1, w, \alpha) = \Rexp(r_0, w, \alpha) \cup \Rexp(r_1, w, \alpha)$,
\item $\Rexp(r_0 \conc r_1, w, \alpha) = \set{(n_0, n_2) \mid \exists n_1 \text{ s.t. }(n_0,n_1)\in \Rexp(r_0, w, \alpha) \text{ and } (n_1, n_2) \in \Rexp(r_1, w, \alpha)}$, and 
\item $\Rexp(r^*, w, \alpha) = \set{(n,n) \mid n\in\nats} \cup \set{(n_0, n_{k+1}) \mid \exists n_1, \ldots, n_{k} \text{ s.t. } (n_j, n_{j+1}) \in \Rexp(r, w, \alpha) \text{ for all } j \le k}$.
\end{itemize}
We write $(w,\alpha) \models \varphi$ for $(w, 0, \alpha) \models \varphi$ and say that $w$ is a model of $\varphi$ with respect to $\alpha$.

\begin{example}\label{ex_pldl} \hfill
\begin{itemize}

\item The formula~$\theta_{\infty p} \coloneq \bbox{\ttrue^*}\ddiamond{\ttrue^*}p$ expresses that $p$ holds true infinitely often. 

\item In general, every $\pltl$ formula~\cite{AlurEtessamiLaTorrePeled01} (and thus every $\ltl$ formula) can be translated into $\pldl$, e.g., $\F_{\le x} p$ is expressible as $\ddiamond{\ttrue^*}_{\le x} p$ and $ p \U q$ as  $\ddiamond{p^*}q$ or $\ddiamond{p^*q}\ttrue$.

\item  The formula~$\bbox{\ttrue^*}(q \rightarrow \ddiamond{(\ttrue \conc\ttrue)^*p})$ requires that every request (a position where $q$ holds) is followed by a response (a position where $p$ holds) after an even number of steps. 

\end{itemize}
\end{example}

As usual for parameterized temporal logics, the use of variables has to be
restricted: bounding diamond and box operators by the same variable leads
to an undecidable satisfiability problem (cp.~\cite{AlurEtessamiLaTorrePeled01}).

\begin{definition}
\label{def_wellformedformula}
A $\pldl$ formula~$\varphi$ is well-formed, if $\vardiamond( \varphi ) \cap \varbox( \varphi ) =
\emptyset$.
\end{definition}
In the following, we only consider well-formed formulas and drop the qualifier \quot{well-formed}.
We consider the following fragments of $\pldl$. Let $\varphi$ be a $\pldl$ formula: $\varphi$ is an $\ldl$ formula~\cite{GiacomoVardi13}, if $\varphi$ is variable-free, $\varphi$ is a $\pldldiamond$ formula, if $\varbox(\varphi) = \emptyset$, and $\varphi$ is a $\pldlbox$ formula, if $\vardiamond(\varphi) = \emptyset$.
Every $\ldl$, $\pldldiamond$, and every $\pldlbox$ formula is well-formed by definition. As satisfaction of $\ldl$ formulas is independent of variable valuations, we write $(w, n) \models \varphi$ and $w \models \varphi$ instead of $(w, n, \alpha) \models \varphi$ and $(w, \alpha) \models \varphi$, respectively, if $\varphi$ is an $\ldl$ formula. 

$\ldl$ is as expressive as $\omega$-regular languages, which can be proven by a straightforward translation of ETL$_f$~\cite{VardiWolper94}, which expresses exactly the $\omega$-regular languages, into $\ldl$. 

\begin{theorem}[\cite{Vardi11}]
For every $\omega$-regular language $L \subseteq (\pow{P})^\omega$ there exists an effectively constructible $\ldl$ formula $\varphi$ such that $L = \set{w \in (\pow{P})^\omega \mid w \models \varphi }$.
\end{theorem}

Note that we define $\pldl$ formulas to be in negation normal form.
Nevertheless, a negation can be pushed to the atomic propositions using dualities allowing us to define the negation of a
formula.

\begin{lemma}
\label{lemma_pldlnegation}
For every $\pldl$ formula~$\varphi$ there exists an efficiently constructible
$\pldl$
formula~$\neg \varphi$ s.t.\
\begin{enumerate}
\item $(w,n,\alpha)\models \varphi$ if and only if $(w,n,
\alpha) \not\models \neg \varphi$,
\item $\card{\neg \varphi} =  \card{\varphi}$.
\item If $\varphi$ is well-formed, then so is $\neg \varphi$.
and vice versa.
\end{enumerate}
\end{lemma}

\begin{proof}
We construct $\neg \varphi$ by structural induction over $\varphi$ using the dualities of the operators:\medskip

\hspace{-.6cm}\begin{minipage}[b]{.4\textwidth}
	\begin{itemize}
		\item $\neg (p) = \neg p$
		\item $\neg (\varphi \wedge \psi) = (\neg \varphi) \vee (\neg \psi)$
		\item $\neg (\ddiamond{r}\varphi) = \bbox{r}\neg \varphi$
		\item $\neg (\ddiamond{r}_{\le x}\varphi) = \bbox{r}_{\le x}\neg \varphi$
	\end{itemize}
	\end{minipage}
\begin{minipage}[b]{.4\textwidth}
	\begin{itemize}
		\item $\neg (\neg p) = p$
		\item $\neg (\varphi \vee \psi) = (\neg \varphi) \wedge (\neg \psi)$
		\item $\neg (\bbox{r}\varphi) = \ddiamond{r}\neg \varphi$
		\item $\neg (\bbox{r}_{\le y}\varphi) = \ddiamond{r}_{\le y}\neg \varphi$
	\end{itemize}
	\end{minipage}\medskip

The latter two claims of Lemma~\ref{lemma_pldlnegation} follow from the definition of $\neg \varphi$ while the first one can be shown by a straightforward structural induction over $\varphi$.
\end{proof}

A simple, but very useful property of $\pldl$ is the monotonicity of the
parameterized operators: increasing (decreasing) the values of parameters bounding diamond (box) operators preserves satisfaction.  

\begin{lemma}
\label{lemma_monotonicity}
Let $\varphi$ be a $\pldl$ formula and let $\alpha$ and
$\beta$ be variable valuations satisfying $\beta ( x) \ge \alpha ( x )$ for
every $x \in \vardiamond( \varphi)$ and $\beta ( y) \le \alpha ( y )$ for every $y \in \varbox( \varphi)$. If $(w, \alpha) \models \varphi$, then $(w, \beta) \models
\varphi$.
\end{lemma}

The previous lemma allows us to eliminate parameterized box operators when asking for the existence of a variable valuation satisfying a formula. 

\begin{lemma}
	\label{lemma_removeboxes}
For every $\pldl$ formula~$\varphi$ there is an efficiently constructible $\pldldiamond$ formula~$\varphi'$ of the same size as $\varphi$ such that
\begin{enumerate}
	
	\item for every $\alpha$ there is an $\alpha'$ such that for all $w$: if $(w, \alpha) \models \varphi$ then $(w, \alpha') \models \varphi'$, and
	
	\item for every $\alpha'$ there is an $\alpha$ such that for all $w$: if $(w, \alpha') \models \varphi'$ then $(w, \alpha) \models \varphi$.	

\end{enumerate}
\end{lemma}
 
\begin{proof}
We construct a single test~$\hat{r}$ such that $\Rexp(r, w, \alpha) \cap \set{ (n,n) \mid n\in\nats} = \Rexp(\hat{r}, w, \alpha)$ for every $w$ and every $\alpha$, which suffices to prove the equivalence of $\bbox{r}_{\le y}\psi$ and $\bbox{\hat{r}}\psi$ provided we have $\alpha(y) = 0$, which is sufficient due to monotonicity. We apply the following rewriting rules (in the given order) to $r$:
\begin{enumerate}
	\item Replace every subexpression of the form~$r'^*$ by $\ttrue?$, until no longer applicable.
	\item Replace every subexpression of the form~$\phi \conc r'$ or $r' \conc \phi$ by $\ffalse?$ and replace every subexpression of the form~$\phi + r'$ or $r' + \phi$ by $r'$, where $\phi$ is a propositional formula,  until no longer applicable.
	\item Replace every subexpression of the form~$\psi_0? + \psi_1?$ by $(\psi_0 \vee \psi_1)?$ and replace every subexpression of the form~$\psi_0?\conc \psi_1?$ by $(\psi_0 \wedge \psi_1)?$, until no longer applicable.
\end{enumerate}
After step~2, $r$ contains no  iterations and no propositional atoms unless the expression itself is one. In the former case, applying the last two rules yields a regular expression which is a single test, which we denote by $\hat{r}$. In the latter case, we define $\hat{r} = \ffalse?$. 

Each rewriting step preserves the intersection~$\Rexp(r, w, \alpha) \cap \set{ (n,n) \mid n\in\nats}$. As $\hat{r}$ is a test, we conclude $\Rexp(r, w, \alpha) \cap \set{ (n,n) \mid n\in\nats} = \Rexp(\hat{r}, w, \alpha)$ for every $w$ and every $\alpha$. Note that $\hat{r}$ can be efficiently computed from $r$ and is of the same size as $r$.
Now, replace every subformula~$\bbox{r}_{\le y}\psi$ of $\varphi$ by $\bbox{\hat{r}}\psi$ and denote the formula obtained by $\varphi'$, which is a $\pldldiamond$ formula that is efficiently constructible and of the same size. 

Given an $\alpha$, we define $\alpha_0$ by $\alpha_0(z) = \alpha(z)$, if $z \in \vardiamond(\varphi)$ and $\alpha_0(z) = 0$ otherwise. If $(w, \alpha) \models \varphi$, then $(w, \alpha_0) \models \varphi$ due to monotonicity. By construction of $\varphi'$, we also have $(w, \alpha_0)\models \varphi'$.
On the other hand, if $(w, \alpha') \models \varphi'$, then $(w, \alpha'_0) \models \varphi'$ as well, where $\alpha'_0$ is defined as above. By construction of $\varphi'$, we conclude $(w, \alpha_0) \models \varphi$. 
\end{proof}

\subsection{The Alternating Color Technique and LDL$\mathbf{_{cp}}$}
\label{subsec_altcolor}
In this subsection, we repeat the alternating color technique, which was introduced by Kupferman et al.\ to solve the model checking and the realizability problem for $\prompt$, amongst others. Let $p \notin P$ be a fresh proposition and define $P' = \pow{P\cup \set{p}}$. We think of words in $(\pow{P'})^\omega$ as colorings of words in $(\pow{P})^\omega$, i.e., $w' \in (\pow{P'})^\omega$ is a coloring of $w \in (\pow{P})^\omega$, if we have ${w_n}' \cap P = w_n$ for every position~$n$. Furthermore, $n$ is a changepoint, if $n= 0$ or if the truth value of $p$ differs at positions~$n-1$ and $n$. A block is a maximal infix that has exactly one changepoint, which is at the first position of the infix. By maximality, this implies that the first position after a block is a changepoint. Let $k\ge 1$. We say that $w'$ is $k$-bounded, if every block has length at most $k$, which implies that $w'$ has infinitely many changepoints. Dually, $w'$ is $k$-spaced, if it has infinitely many changepoints and every block has length at least $k$. 

The alternating color technique replaces a parameterized diamond operator~$\ddiamond{r}_{\le x}\psi$ by an unparameterized one that requires the formula~$\psi$ to be satisfied within at most one color change. To this end, we introduce a changepoint-bounded variant~$\ddiamond{\cdot}_{cp}$ of the diamond operator. Since we need the dual operator~$\bbox{\cdot}_{cp}$ to allow for negation via dualization, we introduce it here as well.
We define
\begin{itemize}
	\item $(w, n, \alpha) \models \ddiamond{r}_{cp}\psi'$ if there exists a $j \in \nats$ s.t.\ $(n, n+j) \in \Rexp(r, w, \alpha)$, $w_n \cdots w_{n+j-1}$ contains at most one changepoint, and $(w, n + j, \alpha) \models \psi$, and

	\item $(w, n, \alpha) \models \bbox{r}_{cp}\psi'$ if for all $j \in \nats$ with $(n, n+j) \in \Rexp(r, w, \alpha)$ and where $w_n \cdots w_{n+j-1}$ contains at most one changepoint we have $(w, n + j, \alpha) \models \psi$.

\end{itemize}

We denote the logic obtained by disallowing parameterized operators, but allowing changepoint-bounded operators, by $\ldlt$. Note that the semantics of $\ldlt$ formulas are independent of variable valuations. Hence, we drop them from our notation for the satisfaction relations~$\models$ and $\Rexp$. Also, Lemma~\ref{lemma_pldlnegation} can be extended to $\ldlt$ by adding the rules $\neg(\ddiamond{r}_{cp} \psi) = \bbox{r}_{cp}\neg \psi$ and $\neg(\bbox{r}_{cp} \psi) = \ddiamond{r}_{cp}\neg \psi$ to the proof.

Now, we are ready to introduce the alternating color technique. Given a $\pldldiamond$ formula~$\varphi$, let $\rel{\varphi}$ be the formula obtained by inductively replacing every subformula~$\ddiamond{r}_{\le x}\psi$ by $\ddiamond{\rel{r}}_{cp}\rel{\psi}$, i.e., we replace the parameterized diamond operator by a changepoint-bounded one. Note that this replacement is also performed in the regular expressions, i.e., $\rel{r}$ is the regular expression obtained by applying the replacement to every test $\psi'?$ in $r$. 

Given a $\pldldiamond$ formula~$\varphi$ let $c(\varphi) = \rel{\varphi} \wedge \theta_{\infty p} \wedge \theta_{\infty \neg p}$ (cf.~Example~\ref{ex_pldl}),
which is an $\ldlt$ formula and only linearly larger than $\varphi$. On $k$-bounded and $k$-spaced colorings of $w$ there is an equivalence between $\varphi$ and $c(\varphi)$. The proof is similar to the original one~\cite{KupfermanPitermanVardi09}.

\begin{lemma}[cp. Lemma~2.1 of \cite{KupfermanPitermanVardi09}]
\label{lemma_altcolor}
Let $\varphi$ be a $\pldldiamond$ formula and let $w \in (\pow{P})^\omega$.
\begin{enumerate}
\item\label{lemma_alternatingcolor_pldltoldl} 
If $(w, \alpha) \models \varphi$, then $w' \models c(\varphi)$ for every $k$-spaced coloring $w'$ of $w$, where $k = \max_{x \in \var(\varphi)}\alpha(x)$.

\item\label{lemma_alternatingcolor_ldldtopldl}
Let $k \in \nats$. If $w'$ is a $k$-bounded coloring of $w$ with $w' \models c(\varphi)$, then $(w, \alpha) \models \varphi$, where $\alpha(x) = 2k$ for every $x$.
\end{enumerate}
\end{lemma}
 
\section{From LDL$\mathbf{_{cp}}$ to Alternating Büchi Automata}
\label{sec_automata}
In this section, we show how to translate $\ldlt$ formulas into alternating B\"uchi word automata of linear size using an inductive bottom-up approach. These automata allow us to use automata-based constructions to solve the model checking and the realizability problem for $\pldl$ via the alternating color technique which links $\pldl$ and $\ldlt$.

An alternating Büchi automaton~$\aut = (Q,\Sigma,q_0,\delta, F)$ consists of a finite set~$Q$ of states, an alphabet~$\Sigma$, an initial state~$q_0 \in Q$, a transition function~$\delta \colon Q \times \Sigma \to \bplus(Q)$, and a set~$F \subseteq Q$ of accepting states. 
Here, $\bplus(Q)$ denotes the set of positive boolean combinations over $Q$, which contains in particular the formulas $\ttrue$ (true) and $\ffalse$ (false).  A run of $\aut$ on $w = w_0 w_1 w_2 \cdots \in \Sigma^\omega$ is a directed graph $\rho = (V, E)$ with $V \subseteq Q \times \nats$ and $((q,n),(q',n')) \in E$ implies $n' = n +1$ such that the following two conditions are satisfied: $(q_0, 0) \in V$ and for all $(q, n) \in V$: $\suc{
\rho}{(q,n)} \models \delta(q, w_n)$. Here $\suc{\rho}{(q,n)}$ denotes the set of successors of $(q,n)$ in $\rho$ projected to $Q$. A run~$\rho$ is accepting if all infinite paths (projected to $Q$) through $\rho$ visit $F$ infinitely often. The language~$L(\aut)$ contains all $w \in \Sigma^\omega$ that have an accepting run of $\aut$.

\begin{theorem}
	\label{theorem_autconstruction}
For every $\ldlt$ formula~$\varphi$, there is an alternating B\"uchi automaton~$\aut_\varphi$ with linearly many states (in $\card{\varphi}$) such that $L(\aut_\varphi) = \set{w \in (\pow{P'})^\omega \mid w \models \varphi }$.
\end{theorem}

To prove the theorem, we inductively construct automata~$\aut_\psi$ for every subformula~$\psi \in \cl(\varphi)$ satisfying $L(\aut_\psi) = \set{w \in (\pow{P'})^\omega \mid w \models \psi}$. The automata for atomic formulas are straightforward and depicted in Figure~\ref{fig_atomicaut}(a) and (b). To improve readability, we allow propositional formulas over $P'$ as transition labels: the formula~$\phi$ stands for all sets $A \in \pow{P'}$ with $A \models \phi$. Furthermore, given automata~$\aut_{\psi_0}$ and $\aut_{\psi_1}$, using a standard construction, we can build the automaton~$\aut_{\psi_0 \vee \psi_1}$ by taking the disjoint union of the two automata, adding a new initial state~$q_0$ with $\delta(q_0, A) = \delta^0(q_0^0, A) \vee \delta^1(q_0^1, A)$. Here, $q_0^i$ is the initial state and $\delta^i$ is the transition function of $\aut_{\psi_i}$. The automaton~$\aut_{\psi_0 \wedge \psi_1}$ is defined similarly, the only difference being $\delta(q_0, A) = \delta^0(q_0^0, A) \wedge \delta^1(q_0^1, A)$.

\begin{figure}[h]
	\begin{center}
		\vspace{-.2cm}
\begin{tikzpicture}[
	every initial by arrow/.style={initial text=,-stealth, thick},
	every state/.style={thick}]

\node			at (-.8,1.5)	{(a)};
\node			at (3.5,1.5)	{(b)};
\node			at (7.5,1.5)	{(c)};

\node[state, initial] (q) at (0,0)	{};
\node[state, accepting]		(a) at (1.5,.75)	{};
\node[state]		(b)	at (1.5,-.75) {};

\path[thick,-stealth]
(q) edge[bend left] node[above]{$\phi$} (a)
(q) edge[bend right] node[below]{$\neg\phi$} (b)
(a) edge[loop right] node[above, yshift =.15cm]{$\, \, \ttrue$} ()
(b) edge[loop right] node[below, yshift =-.15cm]{$\, \, \ttrue$} ();

\node[state, initial] (qp) at (4,0)	{};
\node[state, accepting]		(ap) at (5.5,.75)	{};
\node[state]		(bp)	at (5.5,-.75) {};

\path[thick,-stealth]
(qp) edge[bend left] node[above]{$\neg\phi$} (ap)
(qp) edge[bend right] node[below]{$\phi$} (bp)
(ap) edge[loop right] node[above, yshift =.15cm]{$\, \, \ttrue$} ()
(bp) edge[loop right] node[below, yshift =-.15cm]{$\, \, \ttrue$} ();

\node[state, initial, accepting] (e) at (8, 0) {};
\node[state, accepting] (b) at (10, .75) {};
\node[state, accepting] (y) at (9.5, -.75) {};
\node[state, accepting] (yb) at (12, -.75) {};
\node[state, accepting] (by) at (12.5, .75) {};
\node[state] (s) at (14, 0) {};

\path[thick, -stealth]
(e) edge[bend left] node[above]{$p$} (b)
(e) edge[bend right] node[below]{$\neg p$} (y)
(b) edge[in=210,out=240,loop] node[left]{$\,\, p$} ()
(y) edge[in=30,out=60,loop] node[right,xshift=-.2cm]{$\,\, \neg p$} ()
(b) edge			node[above]{$\neg p$} (by)
(y) edge			node[below]{$p$} (yb)
(yb) edge[in=30,out=60,loop] node[right]{$\,\, p$} ()
(by) edge[in=210,out=240,loop] node[left]{$\,\, \neg p$} ()
(yb) edge[bend right] node[below]{$\neg p$} (s)
(by) edge[bend left] node[above]{$p$} (s)
(s) edge[loop above] node[above] {$\ttrue$} ()
;

\end{tikzpicture}
\vspace{-.3cm}
\caption{The automata~$\aut_{p}$ (a), $\aut_{\neg p}$ (b), and $\aut_{cp}$ (c), which tracks color changepoints.}
\label{fig_atomicaut}
	\end{center}
\end{figure}
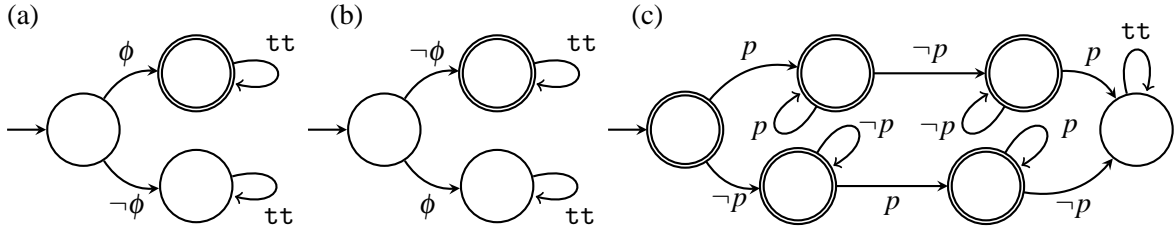

It remains to consider temporal formulas, e.g., $\ddiamond{r}\psi$. First, we turn the regular expression~$r$ into an automaton~$\aut_r$. Recall that tests do not process input letters. Hence, we disregard the tests when defining the transition function, but we label states at which the test has to be executed by this test. We use the Thompson construction~\cite{Thompson68} to turn~$r$ into $\aut_r$, i.e., we obtain an $\epsilon$-NFA. Then, we show how to combine $\aut_r$ with the automaton~$\aut_\psi$ and the automata~$\aut_{\psi_1}, \ldots, \aut_{\psi_k}$, where $\psi_1?, \ldots, \psi_k?$ are the test occurring in $r$. The $\epsilon$-transitions introduced by the Thompson construction are then removed, since alternating automata do not allow them. During this process, we also ensure that the transition relation takes tests into account by introducing universal transitions that lead from a state marked with $\psi_j?$ into the corresponding automaton~$\aut_{\psi_j}$.

Formally, an $\epsilon$-NFA with markings $\aut = (Q, \Sigma, q_0, \delta, C, \marking)$ consists of a finite set~$Q$ of states, an alphabet~$\Sigma$, an initial state~$q_0 \in Q$, a transition function~$\delta \colon Q \times \Sigma\cup\set{\epsilon} \rightarrow \pow{Q}$, a set~$C$ of final states ($C$, since we use them to concatenate automata), and a partial marking function~$\marking$, which assigns to some states~$q \in Q$ an $\ldlt$ formula~$\marking(q)$.
We write $q \xrightarrow{a} q'$, if $q' \in \delta(q, a)$ for $a \in \Sigma \cup \set{\epsilon}$. An $\epsilon$-path~$\pi$ from $q$ to $q'$ in $\aut_r$ is a sequence~$\pi = q_1 \cdots q_k$ of $k \ge 1$ states with $q =q_1 \xrightarrow{\epsilon} \cdots \xrightarrow{\epsilon} q_k = q'$. The set of all $\epsilon$-paths from $q$ to $q'$ is denoted by $\Pi(q, q')$. Let $\marking(\pi) = \set{\marking(q_i) \mid 1 \le i \le k}$  be the set of markings visited by $\pi$.

A run of $\aut$ on $w_0 \cdots w_{n-1} \in \Sigma^*$ is a sequence~$q_0 q_1 \cdots q_n$ such that for every $i$ in the range~$0 \le i \le n-1$ there is a state~$q_{i}'$ reachable from $q_i$ via an $\epsilon$-path~$\pi_{i}$ and with $q_{i+1} \in \delta(q_{i}', w_{i})$. The run is accepting if there is a $q_{n}' \in C$ reachable via an $\epsilon$-path~$\pi_n$ from $q_n$.  This slightly unusual definition (but equivalent to the standard one) simplifies our reasoning below. Also, the definition is oblivious to the marking.

We begin by defining the automaton~$\aut_r$ by induction over the structure of $r$ as depicted in Figure~\ref{fig_autr}. Note that the automata we construct have no outgoing edges leaving the unique final state and that we mark some states with tests~$\psi_j?$ (denoted by labeling states with the test).

\begin{figure}[h!]
	\begin{center}
		\vspace{-.2cm}
\begin{tikzpicture}[
	every initial by arrow/.style={initial text=,-stealth, thick},
	every state/.style={thick}]

\node at (.6,.3) {$\aut_\phi$:};
\node at (.6, -1.2) {$\aut_{\psi?}$:};
\node at (6.4,-.5) {$\aut_{r_0 + r_1}$:};
\node at (-0.7,-2.5) {$\aut_{r_0\conc r_1}$:};
\node at (7, -3.8) {$\aut_{r_0^*}$:};

\node[state, initial] (q) at (2,0)	{};
\node[state, accepting]		(a) at (4,0)	{};

\path[thick,-stealth]
(q) edge node[above]{$\phi$} (a);


\node[state, initial] (qp) at (2,-1.5)	{$\psi?$};
\node[state, accepting]		(ap) at (4,-1.5)	{};

\path[thick,-stealth]
(qp) edge node[above]{$\epsilon$} (ap);


\draw[rounded corners, thick]  (9.3,.7) rectangle (12.7,-.7);
\draw[rounded corners, thick]  (9.3,-1.3) rectangle (12.7,-2.7);

\node		at (11,-.4) {$\aut_{r_0}$};
\node		at (11,-2.4) {$\aut_{r_1}$};

\node[state, initial] (u0) at (8,-1) {};
\node[state] (u1) at (10,0) {};
\node[state] (u2) at (10,-2) {};
\node[state] (u3) at (12,0) {};
\node[state] (u4) at (12,-2) {};
\node[state, accepting] (u5) at (14,-1) {};

\path[thick, -stealth]
(u0) edge[bend left] node[above] {$\epsilon$} (u1)
(u0) edge[bend right] node[below] {$\epsilon$} (u2)
(u3) edge[bend left] node[above] {$\epsilon$} (u5)
(u4) edge[bend right] node[below] {$\epsilon$} (u5);


\draw[rounded corners, thick]  (1.3,-2.5) rectangle (4.7,-3.9);
\draw[rounded corners, thick]  (1.3,-4.5) rectangle (4.7,-5.9);

\node		at (3,-3.6) {$\aut_{r_0}$};
\node		at (3,-5.6) {$\aut_{r_1}$};

\node[state, initial] (u0) at (0,-3.2) {};
\node[state] (u1) at (2,-3.2) {};
\node[state] (u2) at (2,-5.2) {};
\node[state] (u3) at (4,-3.2) {};
\node[state] (u4) at (4,-5.2) {};
\node[state, accepting] (u5) at (6,-5.2) {};

\path[thick, -stealth]
(u0) edge[] node[above] {$\epsilon$} (u1)
(u4) edge[] node[above] {$\epsilon$} (u5);

\draw[thick, rounded corners, -stealth]
(u3.east) -- (5.5,-3.2) -- node[right] {$\epsilon$}  (5.5, -4.2) -- (.5, -4.2) -- (.5, -5.2) -- (u2.west);


\draw[rounded corners, thick]  (9.3,-3.5) rectangle (12.7,-4.9);
\node		at (11,-4.6) {$\aut_{r_0}$};

\node[state, initial] (c0) at (8,-4.2) {};
\node[state] (c1) at (10,-4.2) {};
\node[state] (c2) at (12,-4.2) {};
\node[state] (c3) at (14,-4.2) {};

\path[thick, -stealth]
(c0) edge[bend right = 40] node[below] {$\epsilon$} (c3)
(c2.north) edge[bend right = 60] node[above] {$\epsilon$} (c1.north)
(c0) edge node[above] {$\epsilon$} (c1)
(c2) edge node[above] {$\epsilon$} (c3);

\end{tikzpicture}
\vspace{-.2cm}
\caption{The inductive definition of $\aut_r$ via the Thompson construction.}
\label{fig_autr}
	\end{center}
\end{figure}
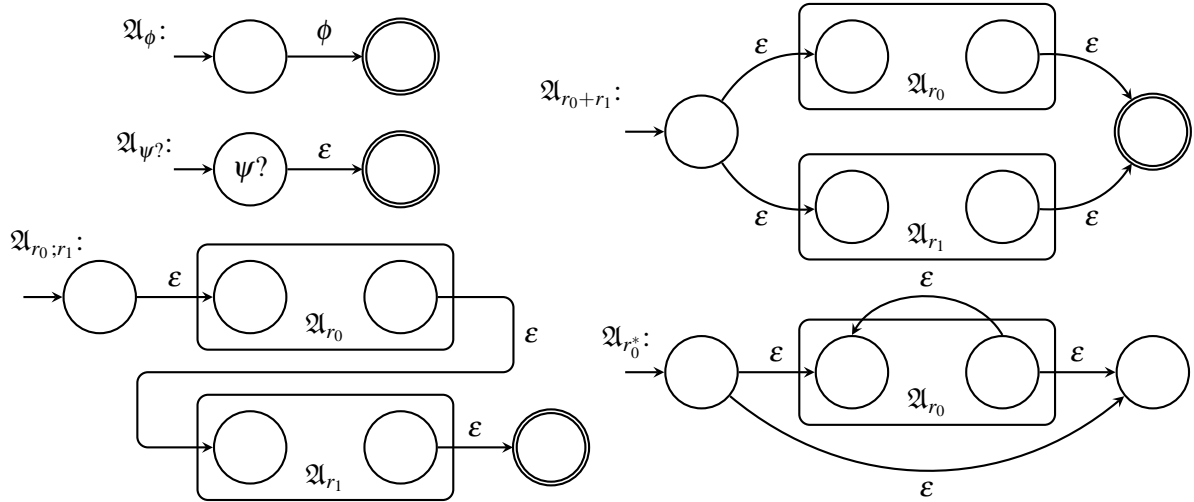

\begin{lemma}
\label{lemma_autrcorrectness}
Let $w = w_0 w_1 w_2 \cdots \in (\pow{P'})^\omega$ and let $w_0 \cdots w_{n-1}$ be a (possibly empty, if $n=0$) prefix of $w$. The following two statements are equivalent:
\begin{enumerate}
	\item $\aut_r$ has an accepting run $q_0 q_1 \cdots q_n$ on $w_0 \cdots w_{n-1}$ with $\epsilon$-paths $\pi_i$ for $i$ in the range~$0 \le i \le n$ such that $w_{i}w_{i+1}w_{i+2} \cdots \models \bigwedge \marking(\pi_i)$ for every $i$.

\item $(0,n) \in \Rexp(r,w)$.

\end{enumerate}
\end{lemma}
 

Fix $\psi$ and $r$ (with tests~$\psi_1?, \ldots, \psi_k?$) and let
$\aut_r = (Q^r, \pow{P'}, q_0^r, \delta^r, C^r, \marking)$, $\aut_{\psi} = (Q', \pow{P'}, q_0', \delta', F')$, and $\aut_{\psi_j} = (Q^j, \pow{P'}, q_0^j, \delta^j, F^j)$ for $j = 1, \ldots, k$ be the corresponding automata, which we assume to have pairwise disjoint sets of states. Next, we show how to construct  $\aut_{\ddiamond{r}\psi}$, $ \aut_{\bbox{r}\psi}$, $ \aut_{\ddiamond{r}_{cp}\psi}$, and $\aut_{\bbox{r}_{cp}\psi}$.

We begin with $\ddiamond{r}\psi$: we define $\aut_{\ddiamond{r}\psi} = (Q^r \cup Q' \cup Q_1 \cup \cdots \cup Q_k, \pow{P'}, q_0^r, \delta, F_1 \cup \cdots \cup F_k)$ with
\[\delta(q, A) = \begin{cases}
	\delta'(q, A) 		&\text{if $q \in Q'$},\\
	\delta^j(q, A) 		&\text{if $q \in Q^j$},\\
	\bigvee_{q' \in Q^r \setminus C^r}\bigvee_{\pi \in \Pi(q,q')} \bigvee_{p \in \delta^r (q', A)} (p \wedge \bigwedge_{\psi_j \in \marking(\pi)} \delta^j(q_0^j, A))&\\
	 \hspace{4.cm}\vee &\text{if $q \in Q^r$}.\\
	\bigvee_{q' \in C^r}\bigvee_{\pi \in \Pi(q,q')} (\delta'(q_0', A) \wedge \bigwedge_{\psi_j \in \marking(\pi)} \delta^j(q_0^j, A))  &\\
	\end{cases}\]
So, $\aut_{\ddiamond{r}\psi}$ is the union of the automata for the regular expression, the tests, and for $\psi$ with a modified transition function. The transitions of the automata~$\aut_{\psi}$ and $\aut_{\psi_j}$ are left unchanged and the transition function for states in $Q^r$ is obtained by removing $\epsilon$-transitions. First consider the upper disjunct: it ranges disjunctively over all non-final states~$p$ that are reachable via an initial $\epsilon$-path and an $A$-transition in the end. To account for the tests visited during the $\epsilon$-path (but not the test at $p$), we add conjunctively transitions that lead into the corresponding automata. The lower disjunct is similar, but ranges over paths that end in a final state. Since we concatenate the automaton~$\aut_r$ with the automaton $\aut_\psi$, all edges leading into final states of $\aut_r$ are rerouted to the initial state of $\aut_\psi$. The tests along the $\epsilon$-path are accounted for as in the first case. Finally, note that $Q^r$ does not contain any (B\"uchi) accepting states, i.e., every accepting run on $w$ has to leave $Q^r$ after a finite number of transitions. Since this is only possible via transitions that would lead $\aut_r$ into a final state, this ensures the existence of a position $n$ such that $(0,n) \in \Rexp(r,w)$.

The definition of $\aut_{\bbox{r}\psi}$ is dual, i.e., we have to use automata~$\aut_{\neg \psi_j}  = (Q^j, \pow{P'}, q_0^j, \delta^j, F^j)$ for $j = 1, \ldots, k$ for the negated tests and $\epsilon$-transitions are removed in a universal manner. Formally, we define $\aut_{\bbox{r}\psi} = (Q^r \cup Q' \cup Q_1 \cup \cdots \cup Q_k, \pow{P'}, q_0^r, \delta, Q^r \cup F_1 \cup \cdots \cup F_k)$ where
\[\delta(q, A) = \begin{cases}
	\delta'(q, A) 		&\text{if $q \in Q'$},\\
	\delta^j(q, A) 		&\text{if $q \in Q^j$},\\
	\bigwedge_{q' \in Q^r \setminus C^r}\bigwedge_{\pi \in \Pi(q,q')} \bigwedge_{p \in \delta^r (q', A)} (p \vee \bigvee_{\psi_j \in \marking(\pi)} \delta^j(q_0^j, A)) &\\
	\hspace{4.2cm}\wedge &\text{if $q \in Q^r$}.\\
	\bigwedge_{q' \in C^r}\bigwedge_{\pi \in \Pi(q,q')} (\delta'(q_0', A) \vee \bigvee_{\psi_j \in \marking(\pi)} \delta^j(q_0^j, A)) &\\
	\end{cases}\]
Note that we add $Q^r$ to the (B\"uchi) accepting states, since a run on $w$ might stay in $Q^r$ forever, as it has to consider all positions $n$ with $(0,n) \in \Rexp(r,w)$.

For the changepoint-bounded operators, we have to modify $\aut_r$ to make it count color changes. Let $\aut_{cp} = (Q^{cp}, \pow{P'}, q_0^{cp}, \delta^{cp}, C^{cp})$ be the DFA depicted in Figure~\ref{fig_atomicaut}(c). We define the product of $\aut_r$ and $\aut_{cp}$ as $\hat{\aut}_r = (\hat{Q}^r, \pow{P'}, \hat{q}_0^r, \hat{\delta}^r, \hat{C}^r, \hat{\marking}) $ where $\hat{Q}^r = Q^r \times Q^{cp}$, $\hat{q}_0^r = (q_0^r, q_0^{cp})$,
\[\delta((q,q'), A) = \begin{cases}
	\set{(p,\delta^{cp}(q',A)) \mid p \in \delta^r(q, A)} &\text{if $A \not= \epsilon$},\\
	\set{(p,q') \mid p \in \delta^r(q, A)} &\text{if $A = \epsilon$},
	\end{cases}\]
$\hat{C}^r = C^r \times C^{cp}$, and $\hat{\marking}(q,q') = \marking(q)$. 
Using this, we define $ \aut_{\ddiamond{r}_{cp}\psi}$ as we defined $\aut_{\ddiamond{r}\psi}$, but using $\hat{\aut}_r$ instead of $\aut_r$. Similarly, $\aut_{\bbox{r}_{cp}\psi}$ is defined as $\aut_{\bbox{r}\psi}$, but using $\hat{\aut}_r$ instead of $\aut_r$.

\begin{proof}[Proof of Theorem~\ref{theorem_autconstruction}]
First, we consider the size of $\aut_\varphi$. Boolean operations add one state while a temporal operator with regular expression~$r$ adds a number of states that is linear in the size of $r$ (which is its length), even when we take the intersection with the automaton checking for color changes. Note that we do not need to complement the automata~$\aut_{\psi_j}$ to obtain $\aut_{\neg \psi_j}$, instead we rely on Lemma~\ref{lemma_pldlnegation}. Hence, the size of $\aut_\varphi$ is linear in the size of $\varphi$.
It remains to prove $L(\aut_\varphi) = \set{w \in (\pow{P'})^\omega \mid w \models \varphi }$ by induction over the structure of $\varphi$. The induction start for atomic formulas and the induction step for disjunction and conjunction are trivial, hence it remains to consider the temporal operators. 

Consider $\ddiamond{r}\psi$. If $w \models \ddiamond{r}\psi$, then there exists a position~$n$ such that $w_n w_{n+1} w_{n+2} \cdots \models \psi$ and $(0,n) \in \Rexp(r,w)$. Hence, there is a run of $\aut_r$ on $w_0 \cdots w_{n-1}$ such that the tests visited during the run are satisfied by the appropriate suffixes of $w$. Thus, applying the induction hypothesis yields accepting runs of the test automata on these suffixes. Furthermore, there is an accepting run of $\aut_\psi$ on $w_n w_{n+1} w_{n+2} \cdots$, again by induction hypothesis. These runs can be \quot{glued} together to build an accepting run of $\aut_{\ddiamond{r}\psi}$ on $w$. 

For the other direction, consider an accepting run~$\rho$ of $\aut_{\ddiamond{r}\psi}$ on $w$. Let $n \ge 0$ be the last level of $\rho$ that contains a state from $Q^r$. Such a level has to exist since states in $Q^r$ are not accepting and they have no incoming edges from states of the automata~$\aut_\psi$ and $\aut_{\psi_j}$, but the initial state of $\aut_{\ddiamond{r}\psi}$ is in $Q^r$. Furthermore, $\aut_{\ddiamond{r}\psi}$ is non-deterministic and complete when restricted to states in $Q^r \setminus C^r$. Hence, we can extract an accepting run of $\aut_r$ from $\rho$ on $w_0 \cdots w_{n-1}$ that satisfies additionally the requirements formulated in Statement~1 of Lemma~\ref{lemma_autrcorrectness}, due to the transitions into the test automata and an application of the induction hypothesis. Hence, we have $(0,n) \in \Rexp(r,w) $. Furthermore, from the remainder of $\rho$ (levels greater or equal to $n$) we can extract an accepting run of $\aut_\psi$ on $w_n w_{n+1} w_{n+2} \cdots$. Hence, $w_n w_{n+1} w_{n+2} \cdots \models \psi$ by induction hypothesis. Altogether, we conclude $w \models \ddiamond{r}\psi$.

The case for $\bbox{r}\psi$ is dual, while the cases for the changepoint-bounded operators~$\ddiamond{r}_{cp}\psi$ and $\bbox{r}_{cp}\psi$ are analogous, using the fact that $\aut_{cp}$ only accepts words which have at most one changepoint.     
\end{proof}

Note that the size of $\aut_\varphi$ is linear in $\card{\varphi}$, but it is not clear that it can be computed in polynomial time in $\card{\varphi}$, since the transition functions of subautomata of the form~$\aut_{\ddiamond{r}\psi}$ contain disjunctions that range over the set of $\epsilon$-paths. Here, it suffices to consider paths that do not contain a state twice, but even this restriction still allows for an exponential number of different paths. Fortunately, we do not need to compute $\aut_\varphi$ in polynomial time. It suffices to do it in polynomial space, which is sufficient for the applications in the next sections, which is clearly possible.

Furthermore, using standard constructions (e.g., \cite{MiyanoH84, Schewe09}), we can turn the alternating B\"uchi automaton~$\aut_\varphi$ into a non-deterministic B\"uchi automaton of exponential size and a deterministic parity automaton\footnote{The states of a parity automaton are colored by $\col\colon Q \rightarrow \nats$. It accepts a word~$w$, if it has a run~$q_0 q_1 q_2 \cdots$ on $w$ such that $\max\set{\col(q) \mid q_i = q \text{ for infinitely many i}} $ is even.} of doubly-exponential size with linearly many colors.

\section{Model Checking}
\label{sec_mc}
In this section, we consider the $\pldl$ model checking problem. A ($P$-labeled) transition system $\sys = (S, s_0, E, \ell)$ consists of a finite set~$S$ of states, an initial state~$s_0$, a (left-)total edge relation~$E \subseteq S \times S$, and a labeling~$\ell \colon S \rightarrow \pow{P}$. An initial path through $\sys$ is a sequence~$\pi = s_0 s_1 s_2\cdots$ of states satisfying $(s_n, s_{n+1}) \in E$ for every $n$. Its trace is defined as $\trace(\pi) = \ell(s_0) \ell(s_1) \ell(s_2) \cdots$. We say that $\sys$ satisfies a $\pldl$ formula~$\varphi$ with respect to a variable valuation~$\alpha$, if we have $(\trace(\pi), \alpha) \models \varphi$ for every initial path~$\pi$ of $\sys$.
The model checking problem asks, given a transition system~$\sys$ and a formula~$\varphi$, to determine whether $\sys$ satisfies $\varphi$ with respect to some variable valuation~$\alpha$.

\begin{theorem}
\label{thm_mc}
The $\pldl$ model checking problem is $\pspace$-complete.
\end{theorem}

To solve the $\pldl$ model checking problem, we first notice that we can restrict ourselves to $\pldldiamond$ formulas. Let $\varphi$ and $\varphi'$ be  due defined as in Lemma~\ref{lemma_removeboxes}. Then, $\sys$ satisfies $\varphi$ with respect to some $\alpha$ if and only if $\sys$ satisfies $\varphi'$ with respect to some $\alpha'$. 

Our algorithm is similar to the one presented for $\prompt$ in \cite{KupfermanPitermanVardi09} and uses the alternating color technique. Recall that $p \notin P$ is the fresh atomic proposition used to specify the coloring and induces the blocks, maximal infixes with its unique changepoint at the first position. Let $G = (V, E, v_0, \ell, F)$  denote a colored B\"uchi graph consisting of a finite directed graph $(V, E)$, an initial vertex~$v_0$, a labeling function~$\ell \colon V \rightarrow \pow{\set{p}}$ labeling vertices by $p$ or not, and a set $F \subseteq V$ of accepting states. A path $v_0 v_1 v_2 \cdots $ through $G$ is pumpable, if all its blocks have at least one state that appears twice in this block. Furthermore, the path is fair, if it visits $F$ infinitely often. The pumpable non-emptiness problem asks, given a colored B\"uchi graph~$G$, whether it has a pumpable fair path starting in the initial state.

\begin{theorem}[\cite{KupfermanPitermanVardi09}]
The pumpable non-emptiness problem for colored B\"uchi graphs is $\nlogspace$-complete and can be solved in linear time.
\end{theorem}

The following lemma reduces the $\pldldiamond$ model checking problem to the pumpable non-emptiness problem for colored B\"uchi graphs of exponential size. Given a non-deterministic B\"uchi automaton~$\aut = (Q, \pow{P \cup \set{p}}, q_0, \Delta, F)$ recognizing the models of $\neg \rel{ \varphi} \wedge \theta_{\infty p} \wedge \theta_{\infty \neg p}$ (note that $\rel{\varphi}$ is negated) and a transition system~$\sys = (S, s_0, E, \ell)$, we define the product~$\aut \times \sys$ to be the colored B\"uchi graph 
\[ \aut \times \sys = (Q \times S \times \pow{\set{p}}, E', (q_0, s_0, \emptyset), \ell', F \times S \times \pow{\set{p}})\]
 where
$((q, s, C),(q', s', C')) \in E'$ if and only if $(s,s')\in E$ and $q' \in \delta(q, \ell(s) \cup C)$, and where $\ell'(q,s,C) = C$.

Each initial path $(q_0, s_0, C_0)(q_1, s_1, C_1)(q_2, s_2, C_2)\cdots$ through the product~$\aut \times \sys$ induces a coloring $(L(s_0) \cup C_0)(L(s_1) \cup C_1)(L(s_2) \cup C_2)\cdots$ of the trace of the path $s_0 s_1 s_2 \cdots$ through $\sys$. Furthermore, $q_0 q_1 q_2 \cdots$ is a run of $\aut$ on the coloring.

\begin{lemma}[cp.\ Lemma 4.2 of \cite{KupfermanPitermanVardi09}]
\label{lemma_pumppath}
$\sys$ does not satisfy $\varphi$ with respect to any $\alpha$ if and only if $\aut \times \sys$ has a pumpable fair path.
\end{lemma}

\begin{proof}
Let $\varphi$ not be satisfied by $\sys$ with respect to any $\alpha$, i.e., for every $\alpha$ there exists an initial path~$\pi$ through $\sys$ such that $(\trace(\pi), \alpha) \not\models \varphi$. Pick $\alpha^*$ such that $\alpha^*(x) = 2\cdot|Q| \cdot |S| +1$ and let $\pi^*$ be the corresponding path. Applying Lemma~\ref{lemma_altcolor}.\ref{lemma_alternatingcolor_ldldtopldl} yields $w \not \models c(\varphi)$ for every $|Q|\cdot|S|$-bounded coloring of $\trace(\pi^*)$. Now, consider the unique $|Q|\cdot|S|$-bounded and $|Q|\cdot|S|$-spaced coloring~$w$ of $\trace(\pi^*)$ that starts with $p$ not holding true in the first position. As argued above, $w \not \models c(\varphi)$, and we have $w \models \theta_{\infty p} \wedge \theta_{\infty \neg p}$, as $w$ is bounded. Hence, $w \models \neg \rel{ \varphi} \wedge \theta_{\infty p} \wedge \theta_{\infty \neg p}$, i.e., there is an accepting run~$q_0 q_1 q_2 \cdots$ of $\aut$ in $w$. This suffices to show that $(q_0, \pi_0, w_0\cap \set{p}) (q_1, \pi_1, w_1\cap \set{p}) (q_1, \pi_1, w_2\cap \set{p}) \cdots$ is a pumpable fair path through $\aut \times \sys$, since every block has length greater than $|Q|\cdot|S|$. This implies the existence of a repeated state in every block, since there are exactly $|Q|\cdot|S|$ vertices of each color.

Now, let $\aut \times \sys$ contain a pumpable fair path $(q_0, s_0, C_0)(q_1, s_1, C_1)(q_2, s_2, C_2)\cdots$, fix some arbitrary $\alpha$, and define $k = \max_{x \in \vardiamond{\varphi}}\alpha(x)$. There is a repetition of a vertex of $\aut \times \sys$ in every block, each of which can be pumped $k$ times. This path is still fair and induces a coloring~$w_k'$ of a trace~$w_k$ of an initial path of $\sys$. Since the run encoded in the first components is an accepting one on $w_k'$, we conclude that the coloring~$w_k'$ satisfies $\neg c(\varphi)$. Furthermore, $w_k'$ is $k$-spaced, since we pumped each repetition $k$ times.

Towards a contradiction assume we have $(w, \alpha) \models \varphi$. 
Applying Lemma~\ref{lemma_altcolor}.\ref{lemma_alternatingcolor_pldltoldl} yields $w' \models c(\varphi)$, which contradicts $\neg c(\varphi)$. Hence, for every $\alpha$ we have constructed a path of $\sys$ whose trace does not satisfy $\varphi$ with respect to $\alpha$, i.e., $\sys$ does not satisfy $\varphi$ with respect to any $\alpha$.
\end{proof}

We can deduce an upper bound on valuations that satisfy a formula in a given transition system.

\begin{corollary}
If there is a variable valuation such that $\sys$ satisfies $\varphi$, then there is also one that is bounded exponentially in $\size{\varphi}$ and linearly in the number of states of $\sys$.	
\end{corollary}

\begin{proof}
Let $\sys$ satisfy $\varphi$ with respect to $\alpha$, but not with the valuation~$\alpha^*$ with $\alpha^*(x) = 2\cdot|Q| \cdot |S| +1$. In the preceding proof, we constructed a pumpable fair path in $\aut \times \sys$ starting from this assumption. This contradicts Lemma~\ref{lemma_pumppath}, since $\sys$ satisfying $\varphi$ with respect to $\alpha$ is equivalent to $\aut \times \sys$ not having a pumpable fair path. Since $2\cdot|Q| \cdot |S| +1$ is exponential in $\size{\varphi}$ and linear in $\card{S}$, the result follows.
\end{proof}

A matching lower bound of $2^n$ can be proven by implementing a binary counter with $n$ bits using a formula of polynomial size in $n$. This holds already true for $\prompt$, as noted in~\cite{KupfermanPitermanVardi09}.

It remains to prove the main result of this section: $\pldl$ model checking is $\pspace$-complete.

\begin{proof}[Proof of Theorem~\ref{thm_mc}]
$\pspace$-hardness follows directly from the $\pspace$-hardness of the $\ltl$ model checking problem~\cite{SistlaClarke85}, as $\ltl$ is a fragment of $\pldl$.

The following is a $\pspace$ algorithm: construct $\aut \times \sys$ and check whether it contains a pumpable fair path, which is correct due to Lemma~\ref{lemma_pumppath}. Since the search for such a path can be implemented on-the-fly without having to construct the full product~\cite{KupfermanPitermanVardi09}, it can be implemented using polynomial space.
\end{proof}

\section{Realizability}
\label{sec_real}
In this section, we consider the realizability problem for $\pldl$. Throughout the section, we fix a partition~$(I, O)$ of the set of atomic propositions~$P$. An instance of the $\pldl$ realizability problem is given by a $\pldl$ formula~$\varphi$ (over $P$) and the problem is to decide whether Player~$O$ has a winning strategy in the following game, played in rounds~$n \in \nats$: in each round~$n$, Player~$I$ picks a subset $i_n \subseteq I$ and then Player~$O$ picks a subset~$o_n \subseteq O$. Player~$O$ wins the play with respect to a variable valuation~$\alpha$, if 
$((i_0 \cup o_0)(i_1 \cup o_1)(i_2 \cup o_2) \cdots, \alpha) \models \varphi$.

Formally, a strategy for Player~$O$ is a mapping~$\sigma\colon (\pow{I})^* \rightarrow \pow{O}$ and a play~$\rho = i_0 o_0 i_1 o_1 i_2 o_2 \cdots $ is consistent with $\sigma$, if we have $o_n = \sigma(i_0 \cdots i_n)$ for every $n$. We call $(i_0 \cup o_0)(i_1 \cup o_1)(i_2 \cup o_2) \cdots$ the outcome of $\rho$, denoted by $\outcome(\rho)$. We say that a strategy~$\sigma$ for Player~$I$ is winning with respect to a variable valuation~$\alpha$, if we have $(\outcome(\rho), \alpha) \models \varphi$ for every play~$\rho$ that is consistent with $\sigma$. The $\pldl$ realizability problem asks for a given $\pldl$ formula~$\varphi$, whether Player~$O$ has a winning strategy with respect to some variable valuation, i.e., there is a single $\alpha$ such that every outcome satisfies $\varphi$ with respect to $\alpha$. If this is the case, then we say that $\sigma$ realizes $\varphi$ and thus that $\varphi$ is realizable.

We show the $\pldl$ realizability problem to be $\twoexp$-complete: hardness follows easily from the $\twoexp$-completeness of the $\ltl$ realizability problem, which is a special case of the $\pldl$ realizability problem. Membership in $\twoexp$ on the other hand is shown by a reduction to the realizability problem for $\omega$-regular specifications. 

It is well-known that $\omega$-regular specifications are realizable by finite-state transducers (if they are realizable at all)~\cite{BuechiLandweber69}. A transducer~$\trans = (Q, \Sigma, \Gamma, q_0, \delta, \tau )$ consists of a finite set~$Q$ of states, an input alphabet~$\Sigma$, an output alphabet~$\Gamma$, an initial state~$q_0$, a transition function~$\delta \colon Q \times \Sigma \rightarrow Q$, and a output function~$\tau \colon Q \rightarrow \Gamma$. The function~$f_{\trans} \colon \Sigma^* \rightarrow \Gamma$ implemented by $\trans$ is defined as $f_{\trans}(w) = \tau(\delta^*(w))$, where $\delta^*$ is defined as usual: $\delta^*(\epsilon) = q_0$ and $\delta^*(wv) = \delta(\delta^*(w),v)$. To implement a strategy by a transducer, we use $\Sigma = \pow{I}$ and $\Gamma = \pow{O}$. Then, we say that the strategy $\sigma = f_\trans$ is finite-state. The size of $\sigma$ is the number of states of $\trans$. The following proof is analogous to the one for $\prompt$~\cite{KupfermanPitermanVardi09}.

\begin{theorem}
\label{thm_real}
The $\pldl$ realizability problem is $\twoexp$-complete.
\end{theorem}

When proving membership in $\twoexp$, we restrict ourselves without loss of generality to $\pldldiamond$ formulas, as this special case is sufficient as shown in Lemma~\ref{lemma_removeboxes}. First, we use the alternating color technique to show that the $\pldldiamond$ realizability problem is reducible to the realizability problem for specifications in $\ldlt$. When considering the $\ldlt$ realizability problem, we add the fresh proposition~$p$ used to specify the coloring to $O$, i.e., Player~$O$ is in charge of determining the color of each position.

\begin{lemma}[cp.\ Lemma 3.1 of \cite{KupfermanPitermanVardi09}]
	\label{lemma_realred}
A $\pldldiamond$ formula~$\varphi$ over $I$ and $O$ is realizable if and only if the $\ldlt$ formula~$c(\varphi)$ over $I$ and $O \cup \set{p}$ is realizable.
\end{lemma}

\begin{proof}
Let $\varphi$ be realizable, i.e., there is a winning strategy~$\sigma\colon (\pow{I})^+ \rightarrow \pow{O} $ for Player~$O$ with respect to some~$\alpha$. Now, consider the strategy~$\sigma'\colon (\pow{I})^+ \rightarrow \pow{O \cup \set{p}}$ defined by
\[\sigma'(i_0 \cdots i_{n-1}) = 
\begin{cases}
\sigma(i_0 \cdots i_{n-1})					&\text{if $n \bmod 2k < k$,}\\
\sigma(i_0 \cdots i_{n-1}) \cup \set{p}		&\text{otherwise,}
\end{cases}\]
where $k = \max_{x \in \vardiamond(\varphi)} \alpha(x)$. We show that $\sigma'$ realizes $c(\varphi)$. To this end, let $\rho' = i_0 o_0 i_1 o_1 i_2 o_2 \cdots$ be a play that is consistent with $\sigma'$. Then, 
$ \rho = i_0 (o_0 \setminus \set{p}) i_1 (o_1 \setminus \set{p}) i_2 (o_2 \setminus \set{p}) \cdots$ is by construction consistent with $\sigma$, i.e., $(\outcome(\rho), \alpha)\models \varphi$. As $\rho'$ is a $k$-spaced $p$-coloring of $\rho$, we deduce $\rho' \models c(\varphi)$ by applying Lemma~\ref{lemma_altcolor}.\ref{lemma_alternatingcolor_pldltoldl}. Hence, $\sigma'$ realizes $c(\varphi)$.

Now, assume $c(\varphi)$ is realized by $\sigma'\colon (\pow{I})^+ \rightarrow \pow{O \cup \set{p}}$, which we can assume to be finite-state, say it is implemented by $\trans$ with $n$ states. We first show that every outcome that is consistent with $\sigma'$ is $n+1$-bounded. Such an outcome satisfies $c(\varphi)$ and has therefore infinitely many changepoints. Now, assume it has a block of length strictly greater than $n+1$, say between changepoints at positions~$i$ and $j$. Let $q_0 q_1 q_2 \cdots$ be the states reached during the run of $\trans$ on the projection of $\rho$ to $\pow{I}$. Then, there are two positions~$i'$ and $j'$ satisfying $i \le i' < j' < j$ in the block such that $q_{i'} = q_{j'}$. Hence, $q_0 \cdots q_{i'-1}(q_{i'} \cdots q_{j'-1})^\omega$ is also a run of $\trans$. However, the output generated by this run has only finitely many changepoints, since the output at the states~$q_{i'}, \ldots, q_{j'-1}$ coincides when restricted to $\set{p}$. This contradicts the fact that $\trans$ implements a winning strategy, which implies in particular that every output has infinitely many changepoints, as required by the conjunct~$\theta_{\infty p} \wedge \theta_{\infty \neg p}$ of $c(\varphi)$. Hence, $\rho$ is $(n+1)$-bounded.

Now, consider the strategy~$\sigma \colon (\pow{I})^+ \rightarrow \pow{O}$ defined by $\sigma(i_0 \cdots i_{n-1}) = \sigma'(i_0 \cdots i_{n-1}) \cap O$.
By definition, for every play~$\rho$ consistent with $\sigma$, there is a $(n+1)$-bounded $p$-coloring of $\rho$ that is consistent with $\sigma'$. Hence, applying Lemma~\ref{lemma_altcolor}.\ref{lemma_alternatingcolor_ldldtopldl} yields $(\rho, \beta) \models \rho$, where $\beta(x) = 2n+2$. Hence, $\sigma$ realizes $\varphi$ with respect to $\beta$. Note that $\sigma$ is also finite-state and of the same size as $\sigma'$.
\end{proof}


\begin{proof}[Proof of Theorem~\ref{thm_real}]
As already mentioned above, $\twoexp$-hardness of the $\ldl$ realizability problem follows immediately from the $\twoexp$-hardness of the $\ltl$ realizability problem~\cite{PnueliRosner89a}, as $\ltl$ is a fragment of $\pldl$.

Now, consider membership and recall that we have argued that it is sufficient to consider $\pldldiamond$. Thus, let $\varphi$ be a $\pldldiamond$ formula. By Lemma~\ref{lemma_realred} we know that it is sufficient to consider the realizability of $c(\varphi)$. Let $\aut = (Q, \pow{I \cup O \cup \set{p}}, q_0, \delta, \col)$ be a deterministic parity automaton recognizing the models of $c(\varphi)$. We turn $\aut$ into a parity game~$\game$ such that Player~$1$ wins $\game$ from some dedicated initial vertex if and only if $c(\varphi)$ is realizable.
To this end, we define the arena~$(V, V_0, V_1, E)$ with $V = Q \cup (Q \times \pow{I})$, $V_0 = Q$, $V_1 = Q \times \pow{I}$, and
	$E = \set{(q, (q,i)) \mid i \subseteq I} \cup \set{(q, i), \delta(q, i \cup o) \mid o \subseteq O \cup \set{p}}$,
i.e., Player~$0$ picks a subset $i \subseteq I$ and Player~$O$ picks a subset~$o \subseteq O$, which in turn triggers the (deterministic) update of the state stored in the vertices. Finally, we define the coloring~$\col_\arena$ of the arena via $\col_\arena(q) = \col_\arena(q,i) =\col(q)$. 

It is straightforward to show that Player~$O$ has a winning strategy from $q_0$ in the parity game~$(\arena, \col_\arena)$ if and only if $c(\varphi)$ (and thus $\varphi$) is realizable. Furthermore, if Player~$1$ has a winning strategy, then $\arena$ can be turned into a transducer implementing a strategy that realizes $c(\varphi)$ using $V$ as set of states. Note that $\card{V}$ is doubly-exponential in $\card{\varphi}$, if we assume that $I$ and $O$ are restricted to propositions appearing in $\varphi$. As the parity game is of doubly-exponential size and has linearly many colors, we can solve it in doubly-exponential time in the size of $\varphi$. This concludes the proof.
\end{proof}

Also, we obtain a doubly-exponential upper bound on a variable valuation that allows to realize a given formula. A matching lower bound already holds for $\pltl$~\cite{Zimmermann13}.

\begin{corollary}
If a $\pldldiamond$ formula~$\varphi$ is realizable with respect to some $\alpha$, then it is realizable with respect to some $\alpha$ that is bounded doubly-exponentially in $\card{\varphi}$. 
\end{corollary}

\begin{proof}
If $\varphi$ is realizable, then so is $c(\varphi)$. Using the construction proving the right-to-left implication of Lemma~\ref{lemma_realred}, we obtain that $\varphi$ is realizable with respect to some $\alpha$ that is bounded by $2n+2$, where $n$ is the size of a transducer implementing the strategy that realizes $c(\varphi)$. We have seen in the proof of Theorem~\ref{thm_real} that the size of such a transducer is at most doubly-exponential in $\card{c(\varphi)}$, which is only linearly larger than $\card{\varphi}$. The result follows. 
\end{proof}

\section{Conclusion}
\label{sec_conc}
We introduced Parametric Linear Dynamic Logic, which extends Linear Dynamic Logic by temporal operators equipped with parameters that bound their scope, similarly to Parametric Linear Temporal Logic, which extends Linear Temporal Logic by parameterized temporal operators. Here, the model checking problem asks for a valuation of the parameters such that the formula is satisfied with respect to this valuation on every path of the transition system. Realizability  is defined in the same spirit. 

We showed $\pldl$ model checking to be complete for $\pspace$ and the realizability problem to be complete for $\twoexp$, just as for $\ltl$. Thus, in a sense, $\pldl$ is not harder than $\ltl$. Finally, we were able to give tight exponential respectively doubly-exponential bounds on the optimal valuations for model checking and realizability.

We did not consider the assume-guarantee model checking problem here, but the algorithm solving the problem for $\prompt$ presented in~\cite{KupfermanPitermanVardi09} should be adaptable to $\pldl$ as well. Another open problem concerns the computation of optimal valuations for $\pldldiamond$ and $\pldlbox$ formulas. By exhaustive search within the bounds mentioned above, one can determine the optima. We expect this to be possible in polynomial space for model checking and in triply exponential space for realizability, which is similar to the situation for $\pltl$~\cite{AlurEtessamiLaTorrePeled01, Zimmermann13}. Note that it is an open question whether optimal valuations for $\pltl$ realizability can be determined in doubly-exponential time.

\bibliographystyle{eptcs}
\bibliography{biblio}


\end{document}